\documentclass[onecolumn,12pt]{IEEEtran} 
\usepackage{color}
\usepackage{amsfonts,color,pslatex}
\usepackage{amssymb,amsmath,latexsym}

\usepackage[para]{threeparttable}

\newcommand{\tr}{{\rm Tr}}
\newcommand{\gf}{{\mathbb{F}}}

\newcommand{\C}{{\mathcal C}}
\newcommand{\D}{{\mathcal D}}
\newcommand{\E}{{\mathcal E}}

\newtheorem{theorem}{Theorem}[section]
\newtheorem{lemma}[theorem]{Lemma}

\newtheorem{example}[theorem]{Example}

\newtheorem{remark}[theorem]{Remark}

\begin{document}

\title{The Weight Enumerator of Three Families of Cyclic Codes\thanks{
Z. Zhou's research was supported by
the Natural Science Foundation of China, Proj. No. 61201243. C. Ding's and M. Xiong's research were supported by
The Hong Kong Research Grants Council, Proj. Nos. 600812 and 606211, respectively.
} }
\author{Zhengchun Zhou,\thanks{Z. Zhou is with the School of Mathematics, Southwest Jiaotong University,
Chengdu, 610031, China (email: zzc@home.swjtu.edu.cn).}
             Aixian Zhang,\thanks{A. Zhang is with the School of Mathematical Sciences, Capital Normal University,
Beijing 100048, P. R. China (email: zhangaixian1008@126.com)}
             Cunsheng Ding\thanks{C. Ding is with the Department of Computer Science
                                                  and Engineering, The Hong Kong University of Science and Technology,
                                                  Clear Water Bay, Kowloon, Hong Kong, China (email: cding@ust.hk).}
             and
             Maosheng Xiong\thanks{M. Xiong is with the Department of Mathematics,
The Hong Kong University of Science and Technology, Clear Water Bay, Kowloon, Hong Kong
(email: mamsxiong@ust.hk). }}

\date{\today}
\maketitle

\begin{abstract}
Cyclic codes are a subclass of linear codes and have wide applications in consumer electronics,
data storage systems, and communication systems due to their efficient encoding and decoding
algorithms. Cyclic codes with many zeros and their dual codes have been a subject of study for
many years. However, their weight distributions are known only for a very small number of cases.
In general the calculation of the weight distribution of cyclic codes is heavily based on the
evaluation of some exponential sums over finite fields.  Very recently, Li, Hu, Feng and Ge studied
a class of $p$-ary cyclic codes of length $p^{2m}-1$, where $p$ is a prime and $m$ is odd.
They determined the weight distribution of this class of  cyclic codes by establishing a connection
between the involved exponential sums with the spectrum of Hermitian forms graphs. In this paper,
this class of $p$-ary cyclic codes is generalized and the weight distribution of the generalized
cyclic codes is settled for both even $m$ and odd $m$ along with the idea of Li, Hu, Feng, and Ge.
The weight distributions of two related families of cyclic codes are also determined.
\end{abstract}


\section{Introduction}

Throughout this paper, let $p$ be a prime and $q$ be a  power of $p$. An $[n,k,d]$ linear code over $\gf_q$
is a $k$-dimensional subspace of $\gf_q^n$ with minimum (Hamming) distance $d$. Let $A_i$ denote the number
of codewords with Hamming weight $i$ in a code $\mathcal{C}$ of length $n$. The weight enumerator of $\mathcal{C}$  is defined by
\begin{eqnarray*}
1+A_1x+A_2x^2+\cdots+A_nx^n.
\end{eqnarray*}
The sequence $(A_1,A_2,\cdots,A_{n})$ is called the weight distribution of the code. Clearly, the weight
distribution gives the minimum distance of the code, and thus the error correcting capability. In addition, the
weight distribution of a code allows the computation of the error probability of error detection and correction
with respect to some error detection and error correction algorithms \cite{Klov}. Thus the study of the weight
distribution of a linear code is important in both theory and applications.

An  $[n,k]$ linear code $\C$ over  $\gf_q$ is called {cyclic} if
$(c_0,c_1, \cdots, c_{n-1}) \in \C$ implies $(c_{n-1}, c_0, c_1, \cdots, c_{n-2})
\in \C$.
By identifying any vector $(c_0,c_1, \cdots, c_{n-1}) \in \gf_q^n$
with
$$
c_0+c_1x+c_2x^2+ \cdots + c_{n-1}x^{n-1} \in \gf_q[x]/(x^n-1),
$$
any code $\C$ of length $n$ over $\gf_q$ corresponds to a subset of $\gf_q[x]/(x^n-1)$.
The linear code $\C$ is cyclic if and only if the corresponding subset in $\gf_q[x]/(x^n-1)$
is an ideal.
It is well known that every ideal of $\gf_q[x]/(x^n-1)$ is principal. Let $\C=\langle g(x) \rangle$,
where $g(x)$ is monic and has the least
degree. Then $g(x)$ is called the  generator polynomial and
$h(x)=(x^n-1)/g(x)$ is referred to as the  parity-check polynomial of
$\C$.  A cyclic code is called irreducible
if its parity-check polynomial is irreducible over $\gf_q$. Otherwise, it is called reducible.

The weight distributions of both irreducible and reducible  cyclic codes have been interesting subjects
of study for many years. For information on the weight distribution of irreducible cyclic codes, the reader
is referred to the recent survey \cite{Ding12}. Information
on the weight distribution of reducible cyclic codes could be found in \cite{YCD}, \cite{Feng07}, \cite{Luo081},
\cite{Luo082}, \cite{Luo10}, \cite{Zeng10}, \cite{Ding11}, \cite{Ma11}, and \cite{Wang12}.

Very recently, Li, Hu, Feng and Ge \cite{Li-Hu-Feng-Ge} studied a class of $p$-ary cyclic codes whose duals
may have arbitrarily many zeros. They determined the weight distribution of this class of  cyclic codes by
establishing a  connection between the involved exponential sums with the spectrum of Hermitian forms graphs.
The objectives of this paper are to generalize this class of $p$-ary cyclic codes and settle the weight distribution
of the class of generalized cyclic codes for both even $m$ and odd $m$ along with the idea of Li, Hu, Feng, and Ge.
The weight distributions of two other related families of cyclic codes are also determined.

This paper is organized as follows. Section \ref{sec-introd} defines the three families of cyclic codes.
Section \ref{sec-prelim} presents results on quadratic forms over finite fields, Cayley graphs and Hermitian
forms graphs which will be needed in the sequel. Sections \ref{sec-wteven} and \ref{sec-wtodd} solve the
weight distribution problem for the three families of cyclic codes. Section \ref{sec-conclusion} summarizes
this paper.

\section{The three families of cyclic codes}\label{sec-introd}

In this section, we introduce the three families of cyclic codes to be studied in the sequel.
Before doing this, we first fix some notations which will be used in the remainder of this paper.
Let  $s=2m$ for a positive integer $m$ and $n=q^s-1$.
Let $\pi$ be a generator of the finite field $\gf_{q^s}$.
Define $t=\lfloor {m\over 2} \rfloor$ and
\begin{eqnarray*}
\Gamma=\left\{
\begin{array}{ll}
\{1,{q^m+1}\}\cup\{q^{2i-1}+1: ~1\leq i\leq t\}, &\textrm{~for~odd~}m \\
\{1\}\cup\{q^{2i-1}+1: ~1\leq i\leq t\}, &\textrm{~for~even~}m.
\end{array} \right.\ \
\end{eqnarray*}

It is easy to prove the following lemma. We omit the proof here.

\begin{lemma}\label{fact_exponents}
With the notation as above,  we have the following conclusions.
\begin{itemize}
\item For any two distinct elements $u$ and $v$ in  $\Gamma$, the two elements $\pi^{-u}$ and
         $\pi^{-v}$ are not conjugate over $\gf_q$.
\item For each $u \in \Gamma$, the smallest positive integer $\ell_u$ such that $q^{\ell_u}u\equiv u~(\bmod~n)$ is equal to
$s$ except for $u=q^m+1$ for which $\ell_{u}=m$.
\end{itemize}
\end{lemma}

For any integer $u$, let $h_u(x)$ denote the minimal polynomial of $\pi^{-u}$ over $\gf_q$. Define then
\begin{eqnarray}\label{eqn_parity_check_h}
h(x)=\prod_{u\in \Gamma}h_u(x).
\end{eqnarray}
By Lemma \ref{fact_exponents}, $h_u(x)$ and $h_v(x)$ are distinct for any pair of distinct $u$ and $v$ in $\Gamma$
and $\textrm{deg}(h_u(x))=s$ for each $u \neq q^m+1$, and $\textrm{deg}(h_{q^m+1}(x))=m$.
It then follows that $h(x)$ divides $x^n-1$ and ${\textrm{deg}(h)}=m^2+2m$ whenever $m$ is even or odd.
Let $\D_{(q,m)}$ denote the cyclic code of length $n$ with parity-check polynomial $h(x)$ of (\ref{eqn_parity_check_h}).
Then the code $\D_{(q,m)}$ has dimension $m^2+2m$ and is the dual of a cyclic code with $\lfloor {m+3\over 2} \rfloor$ zeros.
Let $\E_{(q,m)}$ be the cyclic code  of length $n$ with parity-check polynomial $h(x)(x-1)$. Then $\E_{(q,m)}$
has dimension $m^2+2m+1$.
Let $h'(x)=h(x)/h_1(x)$ and  $\C_{(q,m)}$ be the cyclic code of length $n$ with parity-check polynomial $h'(x)$. Then $\C_{(q,m)}$
has dimension $m^2$.  It is clear that
\begin{align*}
\C_{(q,m)}\subset \D_{(q,m)} \subset \E_{(q,m)}.
\end{align*}

When $q=p$ and $m$ is odd, the codes $\C_{(q,m)}$ were studied and their weight distribution was determined
recently by Li, Hu, Feng, and Ge \cite{Li-Hu-Feng-Ge}. The objective of this paper is to determine the weight distribution
of aforementioned cyclic codes $\C_{(q,m)}$ $\D_{(q,m)}$ and $\E_{(q,m)}$ for any prime power $q$ and any positive integer $m$. Our work was inspired by the idea of  Li, Hu, Feng, and Ge \cite{Li-Hu-Feng-Ge}.

\section{Preliminaries}\label{sec-prelim}

In this section, we present necessary results on quadratic forms over finite fields, Cayley graphs and Hermitian
forms graphs which will be needed in the sequel.
\subsection{Quadratic forms over finite fields}

Identifying $\gf_{q^s}$ with the $s$-dimensional
$\gf_q$-vector space $\gf_q^s$, a function $Q(x)$ from $\gf_{q^s}$ to $\gf_q$ can be regarded as an $s$-variable polynomial over $\gf_q$.
The former is called a quadratic form over $\gf_q$ if the latter is a homogeneous polynomial of degree two in the form
\begin{eqnarray*}
Q(x_1,x_2,\cdots,x_s)=\sum_{1\leq i\leq j\leq s}a_{ij}x_ix_j,
\end{eqnarray*}
where $a_{ij}\in \gf_q$, and we use a basis $\{\beta_1,\beta_2,\cdots,\beta_{s}\}$ of $\gf_{q^s}$ over $\gf_q$ and identify $x=\sum_{i=1}^sx_i\beta_i$ with the vector $(x_1,x_2,\cdots,x_{s})\in \gf_q^s$.
The rank of the
quadratic form $Q(x)$ is defined as the codimension of the $\gf_q$-vector space
\begin{eqnarray*}
V=\{y\in \gf_{q^s}: ~Q(x+y)-Q(x)-Q(y)=0 \textrm{~for~all~}x\in \gf_{q^s}\}.
\end{eqnarray*}
That is $|V|=q^{s-r}$ where $r$ is the rank of $Q(x)$.

Quadratic forms  have been well studied (see \cite{Niddle}, \cite{Klappereven}, \cite{Klapperodd}, for example). Here we follow the treatment in \cite{Klappereven} and \cite{Klapperodd}.
It should be noted that the rank of a quadratic form over $\gf_q$ is the smallest number of variables required to represent the quadratic form, up
to  nonsingular coordinate transformations.
Mathematically, any quadratic
form of rank $r$ can be transferred to  three canonical forms as follows. Throughout this section, let $B_{2j}(x)=x_1x_2+x_3x_4+\cdots+x_{2j-1}x_{2j}$ where $j\geq 0$ is an integer (we assume that $B_0=0$ when $j=0$).
Let $\nu(x)$ be a function over $\gf_q$ defined by $\nu(0)=q-1$  and $\nu(\zeta)=-1$ for any
$\zeta\in \gf_q^*$.

\begin{lemma}\label{Lemma_Klapper1}(\cite{Klappereven})
Let $q$ be even. Then every quadratic form $Q(x)$ of rank $r$ in $s$ variables over $\gf_q$ is equivalent to
one of the following three standard types:

\vspace{2mm}

~~~~Type I: ~~~~~$B_r(x)$,~~~$r$~~even;

\vspace{2mm}

~~~~Type II: ~~~ $B_{r-1}(x)+x_m^2$,~~~$r$~~odd;
\vspace{2mm}

~~~~Type III: ~~ $B_{r-2}(x)+\theta x_{r-1}^2+x_{r-1}x_r+\theta x_m^2$, ~~~$r$~~even;\\
where $\theta$ is a fixed element in $\gf_{q}$ satisfying $\tr_{q/2}(\theta)=1$.
Furthermore, for any $\zeta\in \gf_q$, the number of solutions $x\in \gf_{q^s}$  to the equation $Q(x)=\zeta$ is:

~~~~Type I: ~~~~ $q^{s-1}+\nu(\zeta)q^{s-r/2-1}$;
\vspace{2mm}

~~~~Type II: ~~~ $q^{s-1}$;

\vspace{2mm}
~~~~Type III: ~~ $q^{s-1}-\nu(\zeta)q^{s-r/2-1}$.\\

\end{lemma}

\begin{lemma}\label{Lemma_Klapper2}(\cite{Klapperodd})
Let $q$ be odd. Then every quadratic form $Q(x)$ of rank $r$ in $s$ variables over $\gf_q$ is equivalent to
one of the following three standard types:

\vspace{2mm}
~~~~Type I: ~~~~ $B_r(x)$,~~~$r$~~even;

\vspace{2mm}

~~~~Type II: ~~~ $B_{r-1}(x)+\mu x_m^2$,~~~$r$~~odd;
\vspace{2mm}

~~~~Type III: ~~ $B_{r-2}(x)+x^2_{r-1}-\varsigma x_r^2$,~~~$r$~~even;\\
where $\mu\in \{1, \varsigma\}$ and $\varsigma$ is a fixed nonsquare in $\gf_{q}$.
Furthermore, for any $\zeta\in \gf_q$, the number of solutions $x\in \gf_{q^s}$ to the equation $Q(x)=\zeta$ is:

~~~~Type I: ~~~~ $q^{s-1}+\nu(\zeta)q^{s-r/2-1}$;
\vspace{2mm}

~~~~Type II: ~~~ $q^{s-1}+\eta(\mu \zeta)q^{s-(r+1)/2}$;

\vspace{2mm}
~~~~Type III: ~~ $q^{s-1}-\nu(\zeta)q^{s-r/2-1}$;\\
where $\eta$ is the quadratic (multiplicative) character of $\gf_q$ and $\eta(0)$ is assumed to be 0.
\end{lemma}

It is easy to see from the above classification that a quadratic
form  is equivalent to Type I or Type III if it has even rank and otherwise
is equivalent to Type II. The following result follows directly
from Lemmas \ref{Lemma_Klapper1} and \ref{Lemma_Klapper2}.

\begin{lemma}(\cite{Klappereven}, \cite{Klapperodd})\label{Lemma_Quadratic_modue}
Let $Q(x)$ be a quadratic form from $\gf_{q^s}$ to $\gf_q$ with rank $r$. Define
\begin{eqnarray}\label{Def_T_Q}
T_Q=\sum_{x\in \gf_{q^s}}\omega_p^{\tr_{q/p}(Q(x))},
\end{eqnarray}
where $\omega_p$ is a  primitive $p$-th  root of unity.
Then $|T_Q|=q^{s-r/2}~{\textrm{or}}~0$.
Moreover, if $r$ is even and $T_Q\neq 0$,
then
\begin{eqnarray*}
\sum_{x\in \gf_{q^s}}\omega_p^{\tr_{q/p}(aQ(x))}=\epsilon q^{s-r/2}~\textrm{for~any~}a\in \gf^*_q,
\end{eqnarray*}
here and hereinafter $\epsilon=1$ if $Q(x)$ is equivalent to Type I and $\epsilon=-1$ if $Q(x)$ is equivalent to Type III.
\end{lemma}
\vspace{2mm}

Lemma \ref{Lemma_Quadratic_modue} will be used to
prove that the quadratic forms involved in the next sections must have even rank.
Thus we focus on  quadratic forms with even rank in the sequel.
\vspace{2mm}

\begin{lemma}(\cite{Klappereven},\cite{Klapperodd})\label{Lemma_number_solution}
Let $Q(x)$ be a quadratic form from $\gf_{q^s}$ to $\gf_q$ with even rank $r$.
For each $\beta\in \gf_{q^s}$ and each $\zeta\in \gf_q$, let
$N_{Q,\beta}(\zeta)$ denote the number of solutions $x\in \gf_{q^s}$ to the equation
\begin{eqnarray}\label{eqn_Q_B_zeta}
Q(x)+\tr_{q^s/q}(\beta x)=\zeta.
\end{eqnarray}
Then for each $\zeta\in \gf_q$, there are $q^{s}-q^{r}$ many $\beta$'s such that
$$N_{Q,\beta}(\zeta)=q^{s-1},$$
and there are $q^{r-1}+\epsilon \nu(c) q^{r/2-1}$ many $\beta$'s such that
$$N_{Q,\beta}(\zeta)=q^{s-1}+\epsilon \nu(\zeta+c)q^{s-r/2-1},$$
where $c$ runs through $\gf_{q}$.
\end{lemma}

The following two lemmas will be used to prove the  main results of this paper.
\vspace{2mm}

\begin{lemma}\label{Lemma_distr_main1}
Let $Q(x)$ be a quadratic form from $\gf_{q^s}$ to $\gf_q$ with even rank $r$.
Define
\begin{eqnarray}\label{def_S_Q}
S_{Q}(\beta)=\sum_{a\in \gf^*_q}\sum_{x\in \gf_{q^s}}\omega_p^{\tr_{q/p}(a(Q(x)+\tr_{q^s/q}(\beta x)))}.
\end{eqnarray}
Then,  as $\beta$ runs through $\gf_{q^s}$, the values of  $S_{Q}(\beta)$ have the following distribution:
\begin{eqnarray*}
S_{Q}(\beta)=\left\{
\begin{array}{ll}
0, &q^s-q^{r} \textrm{~times}\\
\epsilon (q-1) q^{s-r/2}, &q^{r-1}+\epsilon (q-1)  q^{r/2-1}\textrm{~times} \\
-\epsilon q^{s-r/2}, &(q^{r-1}-\epsilon q^{r/2-1})(q-1)\textrm{~times} .
\end{array} \right.\ \
\end{eqnarray*}
\end{lemma}

\begin{proof}
According to the definition of $S_{Q}(\beta)$, we have
\begin{eqnarray}\label{eqn_S_N_QB1}
S_{Q}(\beta)&&=\sum_{a\in \gf_q}\sum_{x\in \gf_{q^s}}\omega_p^{\tr_{q/p}(a(Q(x)+\tr_{q^s/q}(\beta x)))}-q^s\nonumber\\
&&=qN_{Q,\beta}(0)-q^s,
\end{eqnarray}
where $N_{Q,\beta}(0)$ is the number of solutions  to Equation (\ref{eqn_Q_B_zeta})
for $\zeta=0$. By Lemma  \ref{Lemma_number_solution} and the definition of $\nu(x)$,  the values of $N_{Q,\beta}(0)$, as $\beta$
runs over $\gf_{q^s}$, have the following distribution
\begin{eqnarray*}
N_{Q,\beta}(0)=\left\{
\begin{array}{ll}
q^{s-1}, &q^s-q^{r} \textrm{~times}\\
q^{s-1}+\epsilon (q-1) q^{s-r/2-1}, &q^{r-1}+\epsilon (q-1)  q^{r/2-1}\textrm{~times} \\
q^{s-1}-\epsilon q^{s-r/2-1}, &(q^{r-1}-\epsilon q^{r/2-1})(q-1)\textrm{~times} .
\end{array} \right.\ \
\end{eqnarray*}
The value distribution of $S_{Q}(\beta)$ then follows from Equation (\ref{eqn_S_N_QB1})
and the value distribution of $N_{Q,\beta}(0)$.
\end{proof}

\vspace{2mm}

\begin{lemma}\label{Lemma_distr_main2}
Let $Q(x)$ be a quadratic form from $\gf_{q^s}$ to $\gf_q$ with even rank $r$. For each $b \in \gf^*_{q}$, define
\begin{eqnarray*}
R_{Q,b}(\beta)=\sum_{a\in \gf^*_q}\sum_{x\in \gf_{q^s}}\omega_p^{\tr_{q/p}(a(Q(x)+\tr_{q^s/q}(\beta x)+b))}.
\end{eqnarray*}
Then, as $\beta$ runs through $\gf_{q^s}$, the values of $R_{Q,b}(\beta)$ for any given $b \in \gf^*_{q}$
have the following distribution
\begin{eqnarray*}
R_{Q,b}(\beta)=\left\{
\begin{array}{ll}
0, &q^s-q^{r} \textrm{~times}\\
\epsilon(q-1) q^{s-r/2}, &q^{r-1}-\epsilon q^{r/2-1}\textrm{~times}\\
-\epsilon q^{s-r/2}, &q^{r}-q^{r-1}+\epsilon q^{r/2-1}\textrm{~times} .
\end{array} \right.\ \
\end{eqnarray*}
\end{lemma}

\begin{proof}
It follows from the definition of  $R_{Q,b}(\beta)$ that
\begin{eqnarray}\label{eqn_S_N_QB2}
R_{Q,b}(\beta)&&=\sum_{a\in \gf_q}\sum_{x\in \gf_{q^s}}\omega_p^{\tr_{q/p}(a(Q(x)+\tr_{q^s/q}(\beta x)+b))}-q^s \nonumber\\
&&=qN_{Q,\beta}(-b)-q^s,
\end{eqnarray}
where $N_{Q,\beta}(-b)$ is the number of solutions to Equation (\ref{eqn_Q_B_zeta})
for $\zeta=-b$. According to Lemma  \ref{Lemma_number_solution} and the definition of $\nu(x)$,  the values of $N_{Q,\beta}(-b)$, as $\beta$
runs through $\gf_{q^s}$, have the following distribution
\begin{eqnarray*}
N_{Q,\beta}(-b)=\left\{
\begin{array}{ll}
q^{s-1}, &q^s-q^{r} \textrm{~times}\\
q^{s-1}+\epsilon(q-1) q^{s-r/2-1}, &q^{r-1}-\epsilon  q^{r/2-1}\textrm{~times} \\
q^{s-1}-\epsilon q^{s-r/2-1}, &q^{r}-q^{r-1}+\epsilon q^{r/2-1}\textrm{~times} .
\end{array} \right.\ \
\end{eqnarray*}
The conclusion  then follows from Equation (\ref{eqn_S_N_QB2})
and the value distribution of $N_{Q,\beta}(-b)$.
\end{proof}

\subsection{Cayley graphs and Hermitian forms graphs}

Let $G$ be a finite group and $D$ be a subset of $G$. A Cayley graph on $G$ with connection
set $D$, denoted by $\textrm{Cay}(G,D)$ is the directed graph with
vertex set $G$ and edge set $\{(x,y): ~xy^{-1}\in D\}$.   Let $\hat{G}$ be the
character group of $G$. For any $\chi\in \hat{G}$, define $\chi(D)=\sum_{x\in D}\chi(x)$.
The spectrum of a Cayley graph is the multiset consisting of the eigenvalues of its adjacency matrix.
It is well known that any two isomorphic Cayley graphs have the same spectrum.
When $G$ is abelian, the spectrum
of the Cayley graph $\textrm{Cay}(G,D)$ are completely determined by
the character sums over $D$.

\begin{lemma}\label{Lemma_graph_Li}(\cite{Li-Hu-Feng-Ge})
Let $G$ be a finite abelian group and $\Gamma=\textrm{Cay}(G,D)$ be a Cayley graph
on $G$ with connection set $D$. Then each character $\chi$ of $G$ corresponds to
an eigenvector for $A(\Gamma)$ with eigenvalue $\chi(D)$ where $A(\Gamma)$ is the adjacency
matrix of $\Gamma$, and the spectrum of $\Gamma$ is exactly the multiset $\{\chi(D): ~\chi\in \hat{G}\}$.
\end{lemma}

We now give a brief introduction to Hermitian forms graphs. Before doing this, we recall
 the notation of Hermitian matrix.
For any $x\in \gf_{q^2}$, its conjugate $\bar{x}$ is defined to be $\bar{x}=x^q$. A matrix
$H$ over $\gf_{q^2}$ is called Hermitian if $H=H^*$ where $H^*$ denotes the conjugate transpose of $H$.
Let $\mathcal{H}$ denote the set of all   Hermitian matrices of order $m$ over $\gf_{q^2}$. Then $\mathcal{H}$ is 
a finite abelian group under the
operation of matrix addition. Let $V=\gf_{q^2}^m$.  A Hermitian forms graph on $V$
is the graph whose vertices are the elements of $\mathcal{H}$ and in which $H_1,H_2\in \mathcal{H}$
are adjacent if and only if ${\textrm{rank}}(H_1-H_2)=1$. Thus, the Hermitian forms graph on $V$
is exactly the Caylay graph ${\textrm{Cay}}(\mathcal{H},\mathcal{K})$ on the abelian group $\mathcal{H}$
where $\mathcal{K}=\{H\in \mathcal{H}: ~{\textrm{rank}}(H)=1\}$.

The following result was stated in \cite{Li-Hu-Feng-Ge} without a detailed proof. For completeness,
we report it here and present a full proof for it.

\begin{lemma}\label{Lemma_size_H}
Let $\mathcal{H}$ and $\mathcal{K}$ be defined as above. Then $|\mathcal{H}|=q^{m^2}$ and $|\mathcal{K}|=(q^{2m}-1)/(q+1)$.
\end{lemma}
\begin{proof}
For any $H=(h_{ij})_{m\times m}\in \mathcal{H}$, it follows from $H=H^*$ that $h_{i,i}=h_{i,i}^q$ for each $1\leq i\leq m$.
This implies that $h_{i,i}\in \gf_{q}$ for each $1\leq i\leq m$. Note that for $i\neq j$, $h_{ij}$
can be any element of $\gf_{q^2}$ and $h_{ij}=h_{ji}$. Thus we have
\begin{eqnarray*}
|\mathcal{H}|=q^{m}q^{2(1+2+\cdots+m-1)}=q^{m^2}.
\end{eqnarray*}

We then prove $|\mathcal{K}|=(q^{2m}-1)/(q+1)$. From the theory of
linear algebra, it is known that a matrix $H\in \mathcal{H}$ has rank $1$ if and only if there is a nonzero
vector ${\bf u}=(u_1,u_2,\cdots,u_{m})\in V$ such that $H={\bf u}^*{\bf u}$, where $V=\gf_{q^2}^m$ and ${\bf u}^*$ is the Hermitian transpose of ${\bf u}$.
Let ${\bf u}$ and  ${\bf{v}}$ be two nonzero vectors in $V$.
Note that ${\bf u}^*{\bf u}={\bf v}^*{\bf v}$ if and only if
${\bf v}=c {\bf u}$ for some $c\in \gf^*_{q^2}$. It then follows that $c^{q+1}=1$.
Thus there are  $q+1$ ${\bf v}$'s such that ${\bf u}^*{\bf u}={\bf v}^*{\bf v}$ for any given
nonzero vector ${\bf u}$ in $V$. This together with the fact that there are totally $q^{2m}-1$ nonzero vectors in $V$
leads to  $|\mathcal{K}|=(q^{2m}-1)/(q+1)$.
\end{proof}

\vspace{2mm}

Hermitian forms graphs have been well studied (see \cite{Brouswer} and \cite{Stanton} for details).
It is known that Hermitian forms graphs on $\gf_{q^2}^m$ form a class of distance regular graphs and have the following
spectrum distribution.

\vspace{2mm}

\begin{lemma}\label{Lemma_graph_spectrum}(\cite{Brouswer})
Let $V=\gf_{q^2}^m$. Then the Hermitian forms graph over $V$, i.e., the Caylay graph ${\textrm{Cay}}(\mathcal{H}, \mathcal{K})$,
has the following eigenvalues
\begin{eqnarray*}
\xi_0={(q^{2m}-1)/ (q+1)}, ~\xi_j=((-q)^{2m-j}-1)/(q+1)~\textrm{~for~}1\leq j\leq m,
\end{eqnarray*}
\end{lemma}
and the frequencies of these eigenvalues are given by
\begin{eqnarray}\label{eqn_f}
f_0=1, ~~~f_j={\left[ {\begin{array}{*{10}c}
   {m}  \\
   j \\
\end{array}}\right ]}_{(-q)}\prod_{\ell=0}^{j-1}((-1)^{m+1}q^m+(-1)^{\ell+1}q^{\ell}) ~\textrm{~for~}1\leq j\leq m,
\end{eqnarray}
where 
\begin{eqnarray*}
{\left[ {\begin{array}{*{10}c}
   {m}  \\
   j \\
\end{array}}\right ]}_{\ell}=\prod_{i=0}^{j-1}{\ell^m-\ell^i\over \ell^j-\ell^i}
\end{eqnarray*}
denotes the Gaussian binomial coefficients with  basis $\ell\neq 1$.

We shall  point out that the results in Lemma \ref{Lemma_graph_spectrum} hold for both even and odd $m$.
Hence Lemma \ref{Lemma_graph_spectrum} is applicable  no matter $m$ is even or odd.

\section{The weight distributions of $\C_{(q,m)}$, $\D_{(q,m)}$, and $\E_{(q,m)}$ for even $m$}\label{sec-wteven}

\subsection{The weight distribution of $\D_{(q,m)}$ for even $m$}

Let $s=2m=4t$. Recall that the parity-check  polynomial of $\D_{(q,m)}$ is given by (\ref{eqn_parity_check_h}).
Using the well-known Delsarte's Theorem \cite{Delsarte}, the code   $\D_{(q,m)}$ can be expressed as
\begin{eqnarray}\label{eqn_D_qm}
\D_{(q,m)}=\{{\bf{c}}_{(\beta,\lambda_1,\lambda_2,\cdots,\lambda_t)}: ~\beta, \lambda_1, \lambda_2 ,\cdots,\lambda_{t}\in \gf_{q^s}\}
\end{eqnarray}
in which the codeword
\begin{eqnarray}\label{eqn_def_c}
{\bf{c}}_{(\beta,\lambda_1,\lambda_2,\cdots,\lambda_t)}=\left(\tr_{q^s/q}(\beta \pi^i)+\sum_{j=1}^{t}\tr_{q^s/q}(\lambda_j \pi^{(q^{2j-1}+1)i})\right)_{i=0}^{q^s-2}.
\end{eqnarray}
For the sake of simplicity, we denote $\Lambda=(\lambda_1,\lambda_2,\cdots,\lambda_t)\in \gf_{q^s}^t$ and
\begin{eqnarray}\label{eqn_Q_lambda}
Q_{\Lambda}(x)=\sum_{j=1}^{t}\tr_{q^s/q}(\lambda_j x^{q^{2j-1}+1}), ~x\in \gf_{q^s}.
\end{eqnarray}
It is easy to check that $Q_{\Lambda}(x)$ is a quadratic form from $\gf_{q^s}$ to $\gf_q$.

In terms of exponential sums,  the Hamming weight of the codeword ${\bf{c}}_{(\beta,\lambda_1,\lambda_2,\cdots,\lambda_t)}$ of (\ref{eqn_def_c})
is equal to
\begin{eqnarray}\label{eqn_Hamming_weight}
&&\textrm{WT}({\bf{c}_{(\beta,\lambda_1,\lambda_2,\cdots,\lambda_t)}})\nonumber\\
&&~~~~~~=q^s-1-|\{0\leq i\leq q^s-2: ~Q_{\Lambda}(\pi^i)+\tr_{q^s/q}(\beta \pi^i)=0\}|\nonumber\\
&&~~~~~~=q^s-1-{1\over q}\sum_{a\in \gf_q}\sum_{x\in \gf^*_{q^s}}\omega_p^{\tr_{q/p}(a(Q_{\Lambda}(x)+\tr_{q^s/q}(\beta x)))}\nonumber\\
&&~~~~~~=q^s-1-{1\over q}\left(\sum_{a\in \gf^*_q}\sum_{x\in \gf_{q^s}}\omega_p^{\tr_{q/p}(a(Q_{\Lambda}(x)+\tr_{q^s/q}(\beta x)))}+q^{s}-q\right)\nonumber\\
&&~~~~~~=q^s-q^{s-1}-{1\over q}S_{Q_{\Lambda}}(\beta),
\end{eqnarray}
where
\begin{eqnarray*}
S_{Q_{\Lambda}}(\beta)=\sum_{a\in \gf^*_q}\sum_{x\in \gf_{q^s}}\omega_p^{\tr_{q/p}(a(Q_{\Lambda}(x)+\tr_{q^s/q}(\beta x)))}.
\end{eqnarray*}

Thus we only need to determine the value distribution of $S_{Q_{\Lambda}}(\beta)$.
Note that the function $Q_{\Lambda}(x)$ of (\ref{eqn_Q_lambda}) is a quadratic form in $s$ variables over $\gf_q$.
According to Lemma \ref{Lemma_distr_main1}, it is sufficient to calculate the rank distribution of $Q_{\Lambda}(x)$ as $\Lambda$ runs through
$\gf_{q^s}^t$. This can be achieved by using the following result.

\begin{lemma}\label{Lemma_graph_rank}
Let $\mathcal{G}=(\gf_{q^{s}}^t,+)$ and $\mathcal{D}$ be a subset of the abelian group $\mathcal{G}$ given by
\begin{eqnarray*}
\mathcal{D}=\left\{\left(x^{q+1},x^{q^3+1},\cdots,x^{q^{m-1}+1}\right): ~x\in \gf_{q^s}^*\right\}.
\end{eqnarray*}
Let $\mathcal{H}$ be the set of all Hermitian matrices of order $m$ over $\gf_{q^2}$ and $\mathcal{K}$
be a subset formed by the matrices in $\mathcal{H}$ with rank 1.  Then the Cayley graph $(\mathcal{G},\mathcal{D})$
is isomorphic to  the Cayley graph $(\mathcal{H},\mathcal{K})$ .
\end{lemma}
\begin{proof}
Recall that $s=2m$. It then follows from the definitions of $\mathcal{G}$ and $\mathcal{D}$ that
$|\mathcal{G}|=q^{st}=q^{m^2}$ and $|\mathcal{D}|=(q^{2m}-1)/(q+1)$.
By Lemma \ref{Lemma_size_H}, we have $|\mathcal{G}|=|\mathcal{H}|$ and $|\mathcal{D}|=|\mathcal{K}|$.
We now find an isomorphism $f$ from  $\mathcal{G}$ to $\mathcal{H}$ satisfying $f(\mathcal{D})=\mathcal{K}$.
Let ${\alpha_1,\alpha_2,\cdots,\alpha_m}$ be a basis of $\gf_{q^{s}}$ over $\gf_{q^2}$. Define
${\alpha}=(\alpha_1,\alpha_2,\cdots,\alpha_m)$ and ${\alpha}^{(i)}=(\alpha^i_1,\alpha^i_2,\cdots,\alpha^i_m)$ for any positive integer $i$.
Then $\alpha^{(i)}\in  \gf_{q^{s}}^m$ for any $i$. Based on $\alpha$, we define a map $f$ from $\mathcal{H}$ to $\mathcal{G}$
by sending $H\in \mathcal{H}$ to
\begin{eqnarray}\label{eqn_isomophorsm}
f(H)=(\alpha^{(q)} H \alpha^T, \alpha^{(q^3)} H \alpha^T, \cdots,  \alpha^{(q^{2t-1})} H \alpha^T),
\end{eqnarray}
where $\alpha^T$ is the transpose of the vector $\alpha$.
Clearly, $f(H_1+H_2)=f(H_1)+f(H_2)$. Thus, $f$ is
a homomorphism. We now prove that $f$ is injective. Suppose that $f(H)=(0,0,\cdots,0)$ for some matrix $H=(h_{jk})_{m\times m}$.
It then follows from (\ref{eqn_isomophorsm}) that $\alpha^{(q^{2i-1})} H \alpha^T=0$ for each $1\leq i\leq t$.
We rewrite $\alpha^{(q^{2i-1})} H \alpha^T$ as
$\alpha^{(q^{2i-1})}(h_{jk})_{m\times m}\alpha^T$. Note that the elements of $\alpha$ belong to $\gf_{q^{2m}}$ and the elements of $H$
belong to $\gf_{q^2}$. By raising $q^{2m-2i+1}$ powers on both sides of $\alpha^{(q^{2i-1})}(h_{jk})_{m\times m}\alpha^T=0$, we have
\begin{eqnarray*}
\alpha(h^{q^{2m-2i+1}}_{jk})_{m\times m}(\alpha^{(q^{2m-2i+1})})^T=\alpha(h^{q}_{jk})_{m\times m}(\alpha^{(q^{2m-2i+1})})^T=0.
\end{eqnarray*}
This together with the fact that $H$ is  Hermitian leads to
\begin{eqnarray*}
\alpha^{(^{q^{2m-2i+1}})}H\alpha^T=0,
\end{eqnarray*}
where $1\leq i\leq t$. Thus we have shown that $\alpha^{(^{q^{2i-1}})}H\alpha^T=0$ for each $1\leq i\leq m$. That is,
\begin{eqnarray}\label{eqn_matrix}
\left(
    \begin{array}{cccc}
     \beta_1 &\beta_2 & \cdots &\beta_m \\
      \beta_1^{q^{2}} & \beta_2^{q^{2}} & \cdots & \beta_m^{q^{2}} \\
      \vdots & \vdots & \ddots & \vdots \\
      \beta_1^{q^{2(m-1)}} &\beta_2^{q^{2(m-1)}}&\cdots &\beta_{m}^{q^{2(m-1)}} \\
    \end{array}
  \right)H \alpha =
          \left(
            \begin{array}{c}
             0 \\
              0 \\
              \vdots \\
              0 \\
            \end{array}
          \right),
\end{eqnarray}
where $\beta_i=\alpha_i^q$ for each $1\leq i\leq m$. Note that $\alpha_1,\alpha_2,\cdots,\alpha_m$ is a basis
of $\gf_{q^{2m}}$ over $\gf_{q^2}$. It is easy to show that $\beta_1,\beta_2,\cdots,\beta_m$
are also a basis of $\gf_{q^{2m}}$ over $\gf_{q^2}$. It then follows from (\ref{eqn_matrix}) that $H\alpha^T=0$  since the determinant
of its  left matrix is not zero thanks to  Corollary 2.38 in \cite{Niddle}. Thus $H$ must be a zero matrix.
This means that $f$ is injective, and further means that $f$ is an isomorphism from $\mathcal{H}$ to $\mathcal{G}$
since $|\mathcal{G}|=|\mathcal{H}|$.

Finally, we show that $f$ maps $\mathcal{K}$ to $\mathcal{D}$. Note that for any $H\in \mathcal{K}$,
there is a vector ${\bf v}=(v_1,v_2,\cdots,v_m)\in \gf^m_{q^2}$ such that
$H={\bf v}{\bf v}^*={\bf v}^T{\bf v}^{(q)}$ where ${\bf v}^*$ is the Hermitian transpose of ${\bf v}$. It then follows that
\begin{eqnarray*}
f(H)&&=(\alpha^{(q)} {\bf v}^T{\bf v}^{(q)} \alpha^T, \alpha^{(q^3)} {\bf v}^T{\bf v}^{(q)} \alpha^T, \cdots,  \alpha^{(q^{2t-1})}{\bf v}^T{\bf v}^{(q)} \alpha^T)\\
&&=(x^{q+1}, x^{q^3+1}, \cdots, x^{q^{m-1}+1})\in \mathcal{D},
\end{eqnarray*}
where $x={\bf v}^{(q)}\alpha^T$.
Since $f$ is injective and $|\mathcal{D}|=|\mathcal{K}|$, we have $f(\mathcal{K})=\mathcal{D}$. This completes the proof.
\end{proof}

\begin{remark}
The idea behind the proof of Lemma \ref{Lemma_graph_rank} is inherited from
the proof of Lemma 4.1 in \cite{Li-Hu-Feng-Ge} where the case $q=p$ and $m$ being odd is settled.
\end{remark}

\begin{lemma}\label{Lemma_T}
Let $Q_{{\Lambda}}(x)$ be the quadratic form defined by (\ref{eqn_Q_lambda}) and
\begin{eqnarray}\label{eqn_TQ_lambda}
T_{Q_\Lambda}=\sum_{x\in \gf_{q^s}}\omega_p^{\tr_{q/p}(Q_{\Lambda}(x))}.
\end{eqnarray}
Then, as $\Lambda$ runs through $\gf_{q^s}^t$, the value distribution of the exponential sum $T_{Q_\Lambda}$
is given by
\begin{eqnarray*}
T_{Q_\Lambda}=(-1)^jq^{2m-j} \textrm{~occurring~}~f_j~\textrm{~times},
\end{eqnarray*}
where $0\leq j\leq m$ and $f_j$ is given by (\ref{eqn_f}).
\end{lemma}

\begin{proof}
Let $\Gamma_1=\textrm{Cay}(\mathcal{G},\mathcal{D})$
and $\Gamma_2=(\mathcal{H},\mathcal{K})$ be the two isomorphic Cayley graphs
mentioned in Lemma \ref{Lemma_graph_rank}.
We now calculate the spectrum of the graph $\Gamma_1$. It
is easy to verify that the character group of $G$ is
\begin{eqnarray*}
\hat{G} =\{\chi_{\Lambda}: ~\Lambda=(\lambda_1,\lambda_2,\cdots,\lambda_t) \in \gf^t_{q^s}\},
\end{eqnarray*}
where
\begin{eqnarray*}
\chi_{\Lambda}(g)=\omega_p^{\sum_{i=1}^{t}\tr_{q^s/p}(\lambda_ig_i)}
\end{eqnarray*}
for any $g=(g_1,g_2,\cdots,g_t)\in \mathcal{G}$. By Lemma \ref{Lemma_graph_Li}, the eigenvalues of $\Gamma_1$
are given by
\begin{eqnarray}\label{eqn_T_eigenvalues}
\chi_{\Lambda}(\mathcal{D})&&=\sum_{g\in \mathcal{D}}\omega_p^{\sum_{i=1}^{t}\tr_{q^s/p}(\lambda_ig_i)}\nonumber\\
&&={1\over q+1}\sum_{x\in \gf^*_{q^s}}\omega_p^{\sum_{i=1}^{t}\tr_{q^s/p}(\lambda_i x^{q^{2i-1}+1})}\nonumber\\
&&={1\over q+1}\sum_{x\in \gf^*_{q^s}}\omega_p^{\tr_{q/p}(\sum_{i=1}^{t}\tr_{q^s/q}(\lambda_i x^{q^{2i-1}+1}))}\nonumber\\
&&={1\over q+1}(T_{Q_{\Lambda}}-1).
\end{eqnarray}
On the other hand, by Lemma \ref{Lemma_graph_rank}, $\Gamma_1$   has the same
spectrum as  $\Gamma_2$.  The value distribution of $T_{Q_{\Lambda}}$ then follows from Equation (\ref{eqn_T_eigenvalues}) and Lemma \ref{Lemma_graph_spectrum}.
\end{proof}

\begin{lemma}\label{Lemma_Rank_dis_Q}
Let  $Q_{\Lambda}(x)$  be the quadratic form defined by (\ref{eqn_Q_lambda}). Then the rank of $Q_{\Lambda}(x)$
is equal to $2j$ for some  $0\leq j\leq m$. Furthermore,
the number of $\Lambda\in \gf_{q^s}^t$ such that the rank of $Q_{\Lambda}(x)$ is $2j$
is equal to $f_j$, where $f_j$ is given by $(\ref{eqn_f})$.
\end{lemma}

\begin{proof}
The conclusion follows directly from Lemmas \ref{Lemma_Quadratic_modue} and \ref{Lemma_T}.
\end{proof}

\vspace{2mm}

\begin{theorem}\label{Thereom_even1}
When $m$ is even, $\D_{(q,m)}$ is a $[q^{2m}-1,m^2+2m,q^{2m-2}(q^2-q-1)]$ cyclic code over $\gf_q$,
and has the weight distribution listed in Table \ref{Table_even1}.
\end{theorem}
\begin{proof}
The weight distribution of $\D_{(q,m)}$ follows directly from Equation (\ref{eqn_Hamming_weight}), and Lemmas \ref{Lemma_Rank_dis_Q} and \ref{Lemma_distr_main1}.
\end{proof}

\vspace{2mm}

\begin{table*}[!t]
\renewcommand{\arraystretch}{1.5}
\centering
\begin{threeparttable}
\caption{Weight Distribution of $\D_{(q,m)}$}\label{Table_even1}
\begin{tabular}{|l|l|}
\hline
 Hamming Weight& Frequency\\
\hline
\hline
  $0$ & $1$ \\
  \hline
 $q^{2m}-q^{2m-1}$ & $q^{2m}-1+\sum_{j=1}^{m}(q^{2m}-q^{2j})f_j$\\
\hline
 $q^{2m}-q^{2m-1}-(-1)^jq^{2m-j-1}(q-1)$& $(q^{2j-1}+(-1)^jq^{j-1}(q-1))f_j$ \\
\hline
 $q^{2m}-q^{2m-1}+(-1)^jq^{2m-j-1}$& $(q^{2j-1}-(-1)^jq^{j-1})(q-1)f_j$ \\
\hline
\end{tabular}
\begin{tablenotes}
 $f_j$ is given by (\ref{eqn_f}), $1\leq j\leq m$
\end{tablenotes}
\end{threeparttable}
\end{table*}

\begin{example}
Let $q=2$ and $m=4$. Then the code $\D_{(2,4)}$ is a $[255,24,64]$  code over $\gf_2$ with the weight enumerator
\begin{eqnarray*}
  &&1+255 x^{64}+35700 x^{96}+856800 x^{112}+5178880 x^{120}+5448075 x^{128}+4569600 x^{136}+\\
  &&666400 x^{144}+21420 x^{160}+85 x^{192}
\end{eqnarray*}
which confirms the weight distribution in Table \ref{Table_even1}.
\end{example}

\begin{example}
Let $q=3$ and $m=2$. Then the code $\D_{(3,2)}$ is a $[80,8,45]$  code over $\gf_3$ with the weight enumerator
\begin{eqnarray*}
1+160 x^{45}+1980 x^{48}+1520 x^{54}+2880 x^{57}+20 x^{72}
\end{eqnarray*}
which confirms the weight distribution in Table \ref{Table_even1}.
\end{example}

\begin{example}
Let $q=2^2$ and $m=2$. Then the code $\D_{(4,2)}$ is a $[255,8,176]$  code over $\gf_{2^2}$ with the weight
enumerator
\begin{eqnarray*}
1+765 x^{176}+15504 x^{180}+12495 x^{192}+36720 x^{196}+51 x^{240}
\end{eqnarray*}
which confirms the weight distribution in Table \ref{Table_even1}.
\end{example}

\subsection{The weight distribution of $\C_{(q,m)}$ for even $m$}

\begin{table*}[!t]
\renewcommand{\arraystretch}{1.5}
\centering
\begin{threeparttable}
\caption{Weight Distribution of $\C_{(q,m)}$}\label{Table_even2}
\begin{tabular}{|l|l|}
\hline
 Hamming Weight& Frequency\\
\hline
\hline
  $0$ & $1$ \\
  \hline
 $q^{2m}-q^{2m-1}-(-1)^jq^{2m-j-1}(q-1)$ & $f_j$\\
\hline
\end{tabular}
\begin{tablenotes}
 $f_j$ is given by (\ref{eqn_f}), $1\leq j\leq m$
\end{tablenotes}
\end{threeparttable}
\end{table*}

\begin{theorem}\label{Thereom_even2}
When $m$ is even, $\C_{(q,m)}$ is a $[q^{2m}-1,m^2,(q^{2m}-q^{2m-1})(1-q^{-2})]$ cyclic code over $\gf_q$,
and has the weight distribution listed in Table \ref{Table_even2}.
\end{theorem}

\begin{proof}
According to the definition of $\C_{(q,m)}$ and Delsarte' s Theorem \cite{Delsarte}, we have
\begin{eqnarray*}
\C_{(q,m)}=\{{\bf{c}}_{(\beta,\lambda_1,\lambda_2,\cdots,\lambda_t)}: ~{\bf{c}}_{(\beta,\lambda_1,\lambda_2,\cdots,\lambda_t)}\in \D_{(q,m)}, ~\beta=0\},
\end{eqnarray*}
where $\D_{(q,m)}$ is given by (\ref{eqn_D_qm}).
Then the weight
of a codeword ${\bf{c}}_{(0,\lambda_1,\lambda_2,\cdots,\lambda_t)}$ in $\C_{(q,m)}$ is
\begin{eqnarray}\label{eqn_Hamming_weight2}
&&\textrm{WT}({\bf{c}_{(0,\lambda_1,\lambda_2,\cdots,\lambda_t)}})\nonumber\\
&&~~~~~=q^s-q^{s-1}-{1\over q}\sum_{a\in \gf_q^*}\sum_{x\in \gf_{q^s}}\omega_p^{\tr_{q/p}(aQ_{\Lambda}(x))}\nonumber \\
&&~~~~~=q^s-q^{s-1}-{q-1\over q}T_{Q_{\Lambda}},
\end{eqnarray}
where $T_{Q_{\Lambda}}$ is given by (\ref{eqn_TQ_lambda}) and the second identity followed from Lemma \ref{Lemma_Quadratic_modue} since $T_{Q_{\Lambda}}\neq 0$ due to Lemma \ref{Lemma_T} .
The desired conclusion
then follows from Equation (\ref{eqn_Hamming_weight2}), Lemmas \ref{Lemma_T} and \ref{Lemma_Rank_dis_Q}.
\end{proof}

\begin{example}
Let $q=2$ and $m=4$. Then the code $\C_{(2,4)}$ is a $[255,16,96]$  code over $\gf_2$ with the weight enumerator
\begin{eqnarray*}
1+ 3570^{96}+38080 x^{120}+23800 x^{144}+85 x^{192}
\end{eqnarray*}
which confirms the weight distribution in Table \ref{Table_even2}.
\end{example}

\begin{example}
Let $q=3$ and $m=2$. Then the code $\C_{(3,2)}$ is a $[80,4,48]$  code over $\gf_3$ with the weight enumerator
\begin{eqnarray*}
1+60 x^{48}+20 x^{72}
\end{eqnarray*}
which confirms the weight distribution in Table \ref{Table_even2}.
\end{example}

\begin{example}
Let $q=2^2$ and $m=2$. Then the code $\C_{(4,2)}$ is a $[255,4,180]$  code over $\gf_{2^2}$ with the weight
enumerator
\begin{eqnarray*}
1+204 x^{180}+51 x^{240}
\end{eqnarray*}
which confirms the weight distribution in Table \ref{Table_even2}.
\end{example}

\subsection{The weight distribution of $\E_{(q,m)}$ for even $m$}

\begin{theorem}\label{Thereom_even3}
When $m$ is even, $\E_{(q,m)}$ is a $[q^{2m}-1,m^2+2m+1,q^{2m-2}(q^2-q-1)-1]$ cyclic code over $\gf_q$,
and has the weight distribution listed in Table III.
\end{theorem}

\begin{proof}
By definition, the code $\E_{(q,m)}$ is given by
\begin{eqnarray*}
\E_{(q,m)}=\{ {\bf c}_{(\beta,\lambda_1,\lambda_2,\cdots,\lambda_t)}+b: ~{\bf c}_{(\beta,\lambda_1,\lambda_2,\cdots,\lambda_t)}\in \D_{(q,m)}, b\in \gf_{q}\},
\end{eqnarray*}
where $\D_{(q,m)}$ is given by (\ref{eqn_D_qm}) and
\begin{eqnarray*}
{\bf c}_{(\beta,\lambda_1,\lambda_2,\cdots,\lambda_t)}+b=(c_0+b,c_1+b,\cdots,c_{q^{2m}-2}+b).
\end{eqnarray*}
Herein $c_i$ denotes the $i$th coordinate of the codeword ${\bf c}_{(\beta,\lambda_1,\lambda_2,\cdots,\lambda_t)}$.
In the following, we distinguish between the cases $b=0$ and $b\neq 0$ to calculate the weight distribution of
the codewords in $\E_{(q,m)}$.

{\textit{Case A}}, where $b=0$: The codewords in this case form exactly the code $\D_{(q,m)}$. Therefore the weight distribution
of the codewords in this case is given by Table I.

{\textit{Case B}}, where $b\neq 0$: In this case, the weight of the codeword ${\bf c}_{(\beta,\lambda_1,\lambda_2,\cdots,\lambda_t)}+b$
is given by
\begin{eqnarray}\label{eqn_Hamming_weight3}
&&\textrm{WT}({\bf{c}_{(\beta,\lambda_1,\lambda_2,\cdots,\lambda_t)}}+b)\nonumber\\
&&~~~~~~=q^s-1-{1\over q}\sum_{a\in \gf_q}\sum_{x\in \gf^*_{q^s}}\omega_p^{\tr_{q/p}(a(Q_{\Lambda}(x)+\tr_{q^s/q}(\beta x)+b))}\nonumber\\
&&~~~~~~=q^s-1-{1\over q}\left(\sum_{a\in \gf^*_q}\sum_{x\in \gf_{q^s}}\omega_p^{\tr_{q/p}(a(Q_{\Lambda}(x)+\tr_{q^s/q}(\beta x)+b))}+q^{s}\right)\nonumber\\
&&~~~~~~=q^s-q^{s-1}-1-{1\over q}R_{Q_{\Lambda}, b}(\beta),
\end{eqnarray}
where
\begin{eqnarray*}
R_{Q_{\Lambda}, b}(\beta)=\sum_{a\in \gf^*_q}\sum_{x\in \gf_{q^s}}\omega_p^{\tr_{q/p}(a(Q_{\Lambda}(x)+\tr_{q^s/q}(\beta x)+b))}.
\end{eqnarray*}
The weight distribution of  the codewords for $b\neq 0$ then follows from Equation (\ref{eqn_Hamming_weight3}), Lemmas \ref{Lemma_distr_main2} and \ref{Lemma_Rank_dis_Q}

Summarizing the results in the two cases  above leads to the weight distribution of $\E_{(q,m)}$ in Table \ref{Table_even3}.
\end{proof}

\begin{table*}[!t]
\renewcommand{\arraystretch}{1.5}
\centering
\begin{threeparttable}
\caption{Weight Distribution of $\E_{(q,m)}$}\label{Table_even3}
\begin{tabular}{|l|l|}
\hline
 Hamming Weight& Frequency\\
\hline
\hline
  $0$ & $1$ \\
  \hline
 $q^{2m}-1$ & $q-1$\\
   \hline
 $q^{2m}-q^{2m-1}$ & $q^{2m}-1+\sum_{j=1}^{m}(q^{2m}-q^{2j})f_j$\\
   \hline
 $q^{2m}-q^{2m-1}-1$ & $q^{2m}-1+\sum_{j=1}^{m}(q^{2m}-q^{2j})(q-1)f_j$\\
\hline
 $q^{2m}-q^{2m-1}-(-1)^jq^{2m-j-1}(q-1)$& $(q^{2j-1}+(-1)^jq^{j-1}(q-1))f_j$ \\
 \hline
 $q^{2m}-q^{2m-1}-1-(-1)^jq^{2m-j-1}(q-1)$& $(q^{2j-1}-(-1)^jq^{j-1})(q-1)f_j$ \\
\hline
 $q^{2m}-q^{2m-1}+(-1)^jq^{2m-j-1}$& $(q^{2j-1}-(-1)^jq^{j-1})(q-1)f_j$ \\
\hline
 $q^{2m}-q^{2m-1}-1+(-1)^jq^{2m-j-1}$& $(q^{2j}-q^{2j-1}+(-1)^jq^{j-1})(q-1)f_j$ \\
\hline
\end{tabular}
\begin{tablenotes}
 $f_j$ is given by (\ref{eqn_f}), $1\leq j\leq m$
\end{tablenotes}
\end{threeparttable}
\end{table*}

\begin{example}
Let $q=2$ and $m=4$. Then the code $\E_{(2,4)}$ is a $[255,25,63]$  code over $\gf_2$ with the weight enumerator
\begin{eqnarray*}
&&1+ 85^{63}+255 x^{64}+21420 x^{95}+35700 x^{96}+666400^{111}+856800 x^{112}+4569600 x^{119}+\\
&&5178880 x^{120}+5448075^{127}+5448075 x^{128}+5178880 x^{135}+4569600 x^{136}+856800^{143}+\\
&&666400 x^{144}+35700 x^{159}+21420 x^{160}+255^{191}+85 x^{192}+ x^{255}
\end{eqnarray*}
which confirms the weight distribution in Table \ref{Table_even3}.
\end{example}

\begin{example}
Let $q=3$ and $m=2$. Then the code $\E_{(3,2)}$ is a $[80,9,44]$  code over $\gf_3$ with the weight enumerator
\begin{eqnarray*}
&&1+200 x^{44}+160 x^{45}+2880 x^{47}+1980 x^{48}+3040 x^{53}+1520 x^{54}+6840 x^{56}+2880 x^{57}+\\
&& 160 x^{71}+20 x^{72}+2 x^{80}
\end{eqnarray*}
which confirms the weight distribution in Table \ref{Table_even3}.
\end{example}

\begin{example}
Let $q=2^2$ and $m=2$. Then the code $\E_{(4,2)}$ is a $[255,9,175]$  code over $\gf_{2^2}$ with the weight
enumerator
\begin{eqnarray*}
&& 1+1683 x^{175}+765 x^{176}+36720 x^{179}+15504 x^{180}+37485 x^{191}+12495 x^{192}+119952 x^{195}+\\
&& 36720 x^{196}+765 x^{239}+51 x^{240}+255 x^3
\end{eqnarray*}
which confirms the weight distribution in Table \ref{Table_even3}.
\end{example}

\section{The weight distributions of $\C_{(q,m)}$, $\D_{(q,m)}$, and $\E_{(q,m)}$ for odd $m$}\label{sec-wtodd}

Let $s=2m=2(2t+1)$. Recall that the parity-check  polynomial of the code $\D_{(q,m)}$ is given by (\ref{eqn_parity_check_h}).
Using Delsarte's Theorem \cite{Delsarte}, the code   $\D_{(q,m)}$ for odd $m$ can be expressed as
\begin{eqnarray*}
\D_{(q,m)}=\{{\bf{c}}_{(\beta,\delta_0,\delta_1,\cdots,\delta_t)}: ~ \delta_0\in \gf_{q^m}, \beta, \delta_1,\cdots,\delta_{t}\in \gf_{q^s}\}
\end{eqnarray*}
in which the codeword
\begin{eqnarray}\label{eqn_def_c_odd}
{\bf{c}}_{(\beta,\delta_0,\delta_1,\cdots,\delta_t)}=\left(\tr_{q^s/q}(\beta \pi^i)+\tr_1^m(\delta_0\pi^{(q^{m}+1)i})+\sum_{j=1}^{t}\tr_{q^s/q}(\delta_j \pi^{(q^{2j-1}+1)i})\right)_{i=0}^{q^s-2}.
\end{eqnarray}
According to their connections with  $\D_{(q,m)}$, the codes $\C_{(q,m)}$ and $\E_{(q,m)}$ are respectively given by:
\begin{eqnarray*}
\C_{(q,m)}=\{{\bf{c}}_{(\beta,\delta_0,\delta_1,\cdots,\delta_t)}: ~{\bf{c}}_{(\beta,\delta_0,\delta_1,\cdots,\delta_t)}\in \D_{(q,m)}, \beta=0\}
\end{eqnarray*}
and
\begin{eqnarray*}
\E_{(q,m)}=\{{\bf{c}}_{(\beta,\delta_0,\delta_1,\cdots,\delta_t)}+b: ~{\bf{c}}_{(\beta,\delta_0,\delta_1,\cdots,\delta_t)}\in \D_{(q,m)}, b\in \gf_q\}.
\end{eqnarray*}

The weight distribution of $\C_{(q,m)}$, $\D_{(q,m)}$, and $\E_{(q,m)}$ for odd $m$ can be determined
in a similar way as what we have shown in Section \ref{sec-wteven} for even $m$.
The only distinction is that the quadratic forms involved have distinct expressions.  The quadratic
form $Q_{\Lambda}(x)$ for $m$ being even is given by (\ref{eqn_Q_lambda}) where $\Lambda\in \gf^t_{q^s}$. Meanwhile the quadratic
form for $m$ being odd is
\begin{eqnarray}\label{eqn_quadratic_P}
P_{\Delta}(x)=\tr_{q^m/q}(\delta_0 x^{q^{m}+1})+\sum_{i=1}^t\tr_{q^s/q}(\delta_i x^{q^{2i-1}+1}), ~
x\in \gf_{q^s},
\end{eqnarray}
where $\Delta=(\delta_0,\delta_1,\cdots,\delta_{t})\in \gf_{q^m}\times \gf_{q^s}^t$. The following lemmas
show that the quadratic form $P_{\Delta}(x)$, as $\Delta$ runs through $\gf_{q^m}\times \gf_{q^s}^t$,
has rank $2j$ for each $0\leq j\leq m$, and the rank distribution  is given by Equation (\ref{eqn_f})
in Lemma \ref{Lemma_graph_spectrum}.

\begin{lemma}\label{Lemma_graph_rank_odd}
Let $\mathcal{G}=(\gf_{q^m}\times \gf_{q^{s}}^t,+)$ and $D$ be a subset of $G$ given by
\begin{eqnarray*}
\mathcal{D}=\left\{\left(x^{q^m+1},x^{q+1},x^{q^3+1},\cdots,x^{q^{m-2}+1}\right): ~x\in \gf_{q^s}^*\right\}.
\end{eqnarray*}
Let $\mathcal{H}$ be the set of all Hermitian matrices of order $m$ over $\gf_{q^2}$ and $\mathcal{K}$
be a subset formed by the matrices in $\mathcal{H}$ with rank 1. Then the Cayley graph $(\mathcal{G},\mathcal{D})$
is isomorphic to the Cayley graph $(\mathcal{H},\mathcal{K})$.
\end{lemma}

\begin{proof}
According to the definitions of $\mathcal{G}$ and $\mathcal{D}$, we have
$|\mathcal{G}|=q^{m^2}$ and $|\mathcal{D}|=(q^{2m}-1)/(q+1)$.
Define a map from $\mathcal{H}$ to $\mathcal{G}$ by
\begin{eqnarray*}
f(H)=(\alpha^{(q^m)} H \alpha^T, \alpha^{(q)} H \alpha^T, \alpha^{(q^3)} H \alpha^T, \cdots,  \alpha^{(q^{2t-1})} H \alpha^T).
\end{eqnarray*}
Note that $|\mathcal{G}|=|\mathcal{H}|$ and $|\mathcal{D}|=|\mathcal{K}|$.
It is  sufficient to prove that $f$ is an isomorphism from $\mathcal{H}$ to $\mathcal{G}$ and sending $\mathcal{K}$ to $\mathcal{D}$.
The proof of this lemma is then similar to that of Lemma \ref{Lemma_graph_rank}.
\end{proof}

\vspace{2mm}

\begin{lemma}\label{Lemma_Rank_dis_P}
Let  $P_{\Delta}(x)$  be the quadratic form defined by (\ref{eqn_quadratic_P}). Then the rank of $P_{\Delta}(x)$
is equal to $2j$ for some  $0\leq j\leq m$. Furthermore,
the number of $\Delta\in \gf_{q^m}\times \gf_{q^s}^t$ such that the rank of $P_{\Delta}(x)$ is $2j$
is equal to $f_j$, where $f_j$ is given by $(\ref{eqn_f})$.
\end{lemma}
\begin{proof}
The proof of this lemma is similar to that of Lemma \ref{Lemma_Rank_dis_Q}.
\end{proof}

\vspace{2mm}

\begin{theorem}\label{Thereom_oddmain}
Let $m$ be odd. Then we have the following.
\begin{itemize}
\item $\D_{(q,m)}$ is a $[q^{2m}-1,m^2+2m,q^{2m-2}(q^2-q-1)]$ cyclic code over $\gf_q$,
and has the weight distribution listed in Table \ref{Table_even1}.

\item $\C_{(q,m)}$ is a $[q^{2m}-1,m^2,(q^{2m}-q^{2m-1})(1-q^{-2})]$ cyclic code over $\gf_q$,
and has the weight distribution listed in Table \ref{Table_even2}.

\item $\E_{(q,m)}$ is a $[q^{2m}-1,m^2+2m+1,q^{2m-2}(q^2-q-1)-1]$ cyclic code over $\gf_q$,
and has the weight distribution listed in Table \ref{Table_even3}.
\end{itemize}
\end{theorem}
\begin{proof}
The proof of this theorem is similar to those of Theorems \ref{Thereom_even1}, \ref{Thereom_even2}, and \ref{Thereom_even3}. The details
of the proof are omitted.
\end{proof}

\begin{example}
Let $q=2$ and $m=5$. Then the code $\D_{(2,5)}$ is a $[1023, 35, 256]$  code over $\gf_{2}$ with the weight
enumerator
\begin{eqnarray*}
&&1+1023 x^{256}+579700 x^{384}+58433760 x^{448}+1765998080 x^{480}+9972695040 x^{496}+11589711243 x^{512}+\\
&&9368289280 x^{528}+1558233600 x^{544}+45448480 x^{576}+347820 x^{640}+341 x^{768}
\end{eqnarray*}
which confirms the weight distribution in Table \ref{Table_even1}.
\end{example}

\begin{example}
Let $q=4$ and $m=3$. Then the code $\C_{(4,3)}$ is a $[4095,9,2880]$  code over $\gf_{2^2}$ with the weight
enumerator
\begin{eqnarray*}
&&1+55692 x^{2880}+205632 x^{3120}+819 x^{3840}
\end{eqnarray*}
which confirms the weight distribution in Table \ref{Table_even2}.
\end{example}

\begin{example}
Let $q=3$ and $m=3$. Then the code $\E_{(3,3)}$ is a $[728,16,404]$  code over $\gf_3$ with the weight enumerator
\begin{eqnarray*}
&&1+1820 x^{404}+1456 x^{405}+262080  x^{431}+180180 x^{432}+13394160 x^{476}+7076160 x^{477}+7339696 x^{485}+\\
&& 3669848 x^{486}+7076160 x^{503}+3159000 x^{504}+622440 x^{512}+262080 x^{513}+1456 x^{647}+182 x^{648}+2 x^{728}
\end{eqnarray*}
which confirms the weight distribution in Table \ref{Table_even3}.
\end{example}

\section{Summary}\label{sec-conclusion}

In this paper, we generalized the $p$-ary cyclic codes described in \cite{Li-Hu-Feng-Ge}, and determined
the weight distribution of this class of cyclic codes $\C_{(q,m)}$ over $\gf_q$ for both even and odd $m$,
where $q$ is any power of $p$.
In addition, we found the weight distributions of two other related families of cyclic codes $\D_{(q,m)}$ and
$\E_{(q,m)}$.

\end{document}